\documentclass[11pt,reqno]{amsart}
\usepackage{a4wide}
\usepackage[utf8]{inputenc}
\usepackage{amsmath, mathscinet, amsthm, amscd, amsfonts, enumitem, amssymb, graphicx, xcolor, mathrsfs, dsfont}
\usepackage[english]{babel}
\usepackage[colorlinks=true,linkcolor=blue]{hyperref}

\makeatletter
\newcommand{\myitem}[2]{
  \item[#1]
  \edef\@currentlabel{#1}
  \label{#2}}
\makeatother

\numberwithin{equation}{section}

\theoremstyle{plain}
\newtheorem{theorem}{Theorem}[section]
\newtheorem{proposition}[theorem]{Proposition}

\newtheorem{corollary}[theorem]{Corollary}

\theoremstyle{definition}

\newtheorem{assumption}[theorem]{Assumption}

\theoremstyle{remark}
\newtheorem{remark}[theorem]{Remark}

\newcommand{\al}{\alpha}

\newcommand{\de}{\delta}
\newcommand{\ep}{\varepsilon}
\newcommand{\ph}{\varphi}

\newcommand{\om}{\omega}

\newcommand{\si}{\sigma}

\newcommand{\la}{\lambda}

\newcommand{\Om}{\Omega}

\renewcommand{\d}{\text{\rm d}}

\newcommand{\E}{\mathbb{E}}
\newcommand{\N}{\mathbb{N}}
\renewcommand{\P}{\mathbb{P}}

\newcommand{\R}{\mathbb{R}}

\newcommand{\Bc}{\mathcal{B}}

\newcommand{\Fc}{\mathcal{F}}

\newcommand{\Xc}{\mathcal{X}}

\begin{document}
	
\title[Limits of Unexpected Losses and Risk Ratios]{Asymptotic Behaviour of Unexpected Losses and Risk Ratios for Co-Monotonic Alternatives}

\author{Max Nendel}
\address{Department of Statistics and Actuarial Science, University of Waterloo}
\email{mnendel@uwaterloo.ca}

\thanks{The author thanks Daniel Bartl, Michael Kupper, Eberhard Mayerhofer, Jan Streicher, and Ruodu Wang for valuable comments related to this work.\ The author gratefully acknowledges financial support from the Natural Sciences and Engineering Research Council of
Canada via Discovery Grant no.\ RGPIN-2025-04219.}
\date{\today}

\begin{abstract}

The aggregation of individual risks in large credit and insurance portfolios is guided by diversification and the law of large numbers, which formalizes the convergence of sample averages to their means.\ At the same time, regulatory capital requirements and insurance premia are designed to provide a capital buffer or risk margin above the mean.\ The resulting excess, given by the difference between the nonlinear valuation of the aggregate loss and the corresponding mean, reflects the idea of protection against unexpected losses in the sense of banking and insurance regulation.\ This paper studies the asymptotic behaviour of this excess for large weighted portfolios.\ The main result shows that, for monotone cash-additive risk measures on Banach-lattice-valued Orlicz spaces, convergence along weighted averages satisfying a weak law of large numbers together with a uniform integrability condition is equivalent to scalar continuity at the origin.\ If the risk measure is positively homogeneous, this continuity condition is automatically satisfied, and we prove that the unexpected losses of large weighted portfolios are of order $o(n\overline\lambda_n)$, where $\overline\lambda_n$ denotes the average weight assigned to the first $n$ random variables. We establish analogous asymptotic results for Choquet insurance premia.\ Finally, we derive risk-ratio limits that quantify the potential underestimation arising when diversified portfolios are compared with co-monotonic alternatives.
\medskip
    
	\smallskip
	
\end{abstract}

\maketitle

\section{Introduction}\label{sec.intro}

Risk aggregation is a central problem in actuarial science, banking, and financial risk management.\ An insurance or credit portfolio consists of many individual positions whose marginal distributions are usually reasonably well understood through statistical estimation procedures, credit ratings, and historical default frequencies.\ In their role as risk takers, financial institutions such as insurers and banks assume these individual risks from policyholders, borrowers, or clients, pool them on their balance sheet, and rely on diversification to reduce idiosyncratic risk.\ The classical pooling principle asserts that, under a suitable law of large numbers, average losses stabilize around a deterministic benchmark, which is typically given by the mean in the case of identically distributed loss profiles.\ Risks are therefore usually considered insurable, or more generally poolable, if the aggregate portfolio satisfies a weak law of large numbers and if the individual losses satisfy a uniform integrability condition excluding excessively heavy tails, often expressed through moment constraints.\ However, the law-of-large-numbers requirement concerns the joint distribution of the losses, including their dependence structure, and not merely their marginal distributions.\ This makes the condition difficult to verify in practice, especially for large portfolios with heterogeneous exposures, and leaves its verification particularly susceptible to model uncertainty in the dependence structure.

To account for model uncertainty and adverse deviations in the joint loss distribution, regulatory capital frameworks in banking and insurance require financial institutions to assess aggregate losses through nonlinear risk measures, which are used to determine capital buffers against unexpected losses.\ In the Basel credit-risk framework, unexpected losses are specified in the Internal-Ratings-Based (IRB) Approach.\ The Basel IRB formula, from which Risk-Weighted Assets (RWAs) are obtained through the Basel conversion factor, is calibrated at a 99.9\% confidence level, and its core capital component can be interpreted as the difference between a VaR-type conditional loss amount and expected loss, see \cite{BCBS2005IRBRiskWeights}.\ In insurance regulation, Solvency II requires the Solvency Capital Requirement to cover unexpected losses and calibrates it as a one-year Value at Risk of basic own funds at confidence level 99.5\% \cite[Art.~101(3)]{SolvencyIIArt101}.\ These regulatory frameworks motivate the study of unexpected losses of the form
\begin{equation}\label{eq.uel.intro}
R\bigg(\sum_{i=1}^n\lambda_iX_i\bigg)-\sum_{i=1}^n\lambda_i m_i\quad \text{for a large portfolio size }n\in \N,
\end{equation}
where $X_1,\ldots,X_n$ are individual losses with expected losses $m_1,\ldots, m_n$, $\overline \la_n:=\frac{1}n\sum_{i=1}^n\la_i$ with portfolio weights $\lambda_1,\ldots,\lambda_n>0$, and $R$ is a suitable risk measure, cf.\ \cite[Chapter 4]{FoellmerSchied2025}.\ In this paper, we consider individual losses that take values in a general Banach lattice $B$ and assume that the weighted losses are insurable in the sense that
\begin{enumerate}
\item[(i)] the weak law of large numbers
\begin{equation}\label{eq.wlln.intro}
\lim_{n\to \infty}\P\bigg(\bigg\|\frac1{n\overline{\la}_n}\sum_{i=1}^n \la_i(X_i-m_i)\bigg\|>\ep\bigg)= 0\quad \text{for all }\ep>0
\end{equation} 
with $(m_n)_{n\in \N}\subset B$ is satisfied and 
\item[(ii)] the sequence $$ \Bigg(\Phi\bigg(\bigg\|\frac1{n\overline\la_n}\sum_{i=1}^n \la_i(X_i-m_i)\bigg\|\bigg)\Bigg)_{n\in \N}$$
 is uniformly integrable, where $\Phi$ is a Young function that satisfies the $\Delta_2$-condition, cf.\ \cite[Definition 1, p.\ 22]{RaoRen1991}.
\end{enumerate}
Moreover, we consider risk measures $R\colon L^\Phi(\Omega;B)\to B$ that are maps defined on the Orlicz space $L^\Phi(\Omega;B)$ such that
\begin{enumerate}
\item[(R1)] $R(X_1)\leq R(X_2)$ for all $X_1,X_2\in L^\Phi(\Om;B)$ with $X_1\leq X_2$ $\P$-a.s.,
\item[(R2)] $R(X+m)=R(X)+m$ for all $X\in L^\Phi(\Om;B)$ and $m\in B$.
\end{enumerate}
We point out that we do not exclude the case $\Phi\equiv 0$.\ In this case, $L^\Phi(\Om;B)$ becomes the space of all $B$-valued random variables and condition (ii) is automatically satisfied.\ For a systematic treatment of risk measures on Orlicz hearts, we refer to \cite{CheriditoLi2008,CheriditoLi2009}. Under the imposed $\Delta_2$-condition, the Orlicz space $L^\Phi(\Om;B)$ coincides with the Orlicz heart, so that the Banach-space-valued Orlicz setting considered here is compatible with the framework in \cite{CheriditoLi2008,CheriditoLi2009}.

The first main contribution is Theorem \ref{thm.main.orlicz}, which states that
\begin{equation}\label{eq.thm.1.intro}
\lim_{n\to \infty} \bigg\|R\bigg(\frac1{n\overline\la_n}\sum_{i=1}^n \la_i X_i\bigg)-\frac1{n\overline\la_n}\sum_{i=1}^n \la_im_i\bigg\|=0
\end{equation}
for every sequence of random variables $(X_n)_{n\in \N}\subset L^\Phi(\Om;B)$ that satisfies (i) and (ii) with $(m_n)_{n\in \N}\subset B$ if and only if
   \begin{equation}\label{eq.cont.intro}   
   \lim_{\la\to 0} R(\la X)=0\quad \text{for every }X\in L^\Phi(\Om;B_+), 
   \end{equation}
   where $B_+$ denotes the positive cone in $B$.\
We stress that we do not require the risk measure to satisfy structural properties other than (R1) and (R2).\ In particular, we neither assume $R$ to be convex nor positively homogeneous.\ Clearly, if $R$ is positively homogeneous, i.e.,
\[
R(\la X)=\la R(X)\quad \text{for all }\la>0\text{ and }X\in L^\Phi(\Om;B_+),
\]
condition \eqref{eq.cont.intro} is automatically satisfied.\ In this case, Theorem \ref{thm.main.orlicz} asserts that
\begin{equation}\label{eq.smallo.intro}
\lim_{n\to \infty} \frac1{n\overline\la_n}\bigg\|R\bigg(\sum_{i=1}^n \la_i X_i\bigg)-\sum_{i=1}^n \la_im_i\bigg\|=0,
\end{equation}
i.e., the unexpected losses described in \eqref{eq.uel.intro} are $o(n\overline\la_n)$ as $n\to \infty$.\ The proof of Theorem \ref{thm.main.orlicz} uses an abstract characterization of continuity at the origin for monotone operators on $L^0(\Om;B)$, see Theorem \ref{thm.cont.lattice.L0} and Theorem \ref{thm.cont.monotone.L0} in Section \ref{sec.continuity}, which are of independent interest.\ These results are related to the classical Namioka--Klee theorem, which asserts automatic boundedness of positive linear maps between suitable ordered spaces, and to its extensions for convex risk measures, such as \cite{BF2009}. They are also closely related in spirit to Borwein's automatic-continuity and openness results for convex relations \cite{Borwein1987}.  

We also provide a necessary and sufficient condition for \eqref{eq.thm.1.intro} to hold for sequences of uniformly bounded real-valued random variables that satisfy the weak law of large numbers \eqref{eq.wlln.intro} and risk measures defined on $L^\infty(\Om)$, cf.\ Theorem \ref{thm.main.infty}.\ In this case, the continuity in \eqref{eq.cont.intro} has to be replaced by continuity from above and below, cf.\ \cite[Lemma 4.21 and Theorem 4.22]{FoellmerSchied2025}.\ In this context, we also refer to \cite{Delbaen2021LawLargeNumbersRiskMeasures} for a statistical law of large numbers for law-invariant risk measures applied to empirical distributions. 

The second main contribution concerns insurance premia of distortion or Choquet type, which have a long tradition in economic theory and actuarial science, cf.\ \cite{Yaari1987} and \cite{WYP1997}, which provide an axiomatic characterization for preferences and insurance prices in the form of distorted survival probabilities and Choquet integrals.\ A key property of insurance premia based on Choquet integrals is co-monotonic additivity, which represents additivity for risks with perfectly positive dependence, see \cite[Theorem 4.94]{FoellmerSchied2025}.\ We refer to \cite{DenuitEtAl2005,DhaeneEtAl2002Applications,DhaeneEtAl2002Theory} for a detailed discussion of co-monotonicity in the context of actuarial science and finance.\ Section \ref{sec.insurance} shows that, for nonnegative losses, a distortion premium
\begin{equation}\label{eq.distortion.intro}
H(X)=\int_0^\infty h\big(\mathbb{P}(X>x)\big)\,\mathrm{d}x
\end{equation}
exhibits the asymptotic behaviour \eqref{eq.smallo.intro} when $R$ is replaced with $H$ and $B=\R$, provided the weighted weak law of large numbers \eqref{eq.wlln.intro} together with a moment condition of the form
\begin{equation}\label{eq.moment.intro}
\sup_{n\in \N} \|X_n\|_p<\infty
\end{equation}
 is satisfied and the distortion function $h\colon [0,1]\to [0,1]$ fulfills $\int_1^\infty h(\frac1{x^p})<\infty$, cf.\ Theorem \ref{thm.main.insurance}.\ The key difference in view of Theorem \ref{thm.main.orlicz} is that Theorem \ref{thm.main.insurance} replaces the uniform integrability in condition (ii) with the uniform moment bound \eqref{eq.moment.intro}, which is not enough to ensure that the sequence $(|X_n|^p)_{n\in \N}\subset L^p(\Om)$ is uniformly integrable, so that Theorem \ref{thm.main.orlicz} does not apply.\ A related yet again different result is the strong law of large numbers for bounded pairwise independent and identically distributed random variables under totally monotone capacities on Polish spaces in \cite{Maccheroni2005}.

The third contribution is related to risk aggregation under dependence uncertainty.\ A central question in this area is how large the risk of an aggregate portfolio can be when the marginal distributions of the individual risks are known, but the dependence structure is not. A substantial body of literature has developed around this topic, starting with the early contributions \cite{Makarov1982} and \cite{Rueschendorf1982}, which provide sharp bounds for the Value at Risk of the sum of two random variables with given marginals.\ The recent works \cite{zbMATH07710962, DeVecchi2026} show that even an arbitrarily small amount of positive dependence, if it cannot be excluded, may lead to perfectly correlated tails beyond a threshold and may make the tail risk of an aggregate coincide with that of a perfectly dependent portfolio.\ Inspired by these works, we study the asymptotic behaviour of the ratio
\[
\frac{R\big(\sum_{i=1}^n\la_i Z_i\big)}{R\big(\sum_{i=1}^n\la_i X_i\big)},
\]
where $(X_n)_{n\in \N}$ satisfies (i) and (ii) with $m_n=\E[X_1]>0$ for all $n\in \N$ and $(Z_n)_{n\in \N}$ behaves co-monotonically with $R(Z_n)=R(Z_1)$ for all $n\in \N$.\ If $R\colon L^p(\Om)\to \R$ is co-monotonically additive, we find that
\[
\lim_{n\to \infty} \frac{R\big(\sum_{i=1}^n\la_i Z_i\big)}{R\big(\sum_{i=1}^n\la_i X_i\big)}=\frac{R(Z_1)}{\E[X_1]},
\]
see Theorem \ref{thm.asymptotic.risk.ratio}.\ The resulting limit can be interpreted as a factor quantifying the potential risk underestimation in the presence of dependence uncertainty as the portfolio size tends to infinity.\ A similar ratio can be obtained if $R$ is replaced by a distortion premium of the form \eqref{eq.distortion.intro}, see Theorem \ref{thm.asymptotic.ins.ratio}.\ If $X_n$ and $Z_n$ are Bernoulli or exponentially distributed for all $n\in \N$, we obtain an explicit representation of the right-hand side for the Value at Risk and the Expected Shortfall, see Remark \ref{rem.risk.ratio}.

In a last step, we consider the worst-case aggregate risk
\begin{equation}\label{eq.wc.risk.intro}
R_n^{\rm WC}(Z):=\sup_{Y_1\sim Z,\ldots, Y_n\sim Z} R\bigg(\sum_{i=1}^n\la_i Y_i\bigg)
\end{equation}
for a coherent risk measure $R\colon L^p(\Om)\to \R$, $Z\in L^p(\Om)$, and $n\in \N$. Risk aggregation problems of the form \eqref{eq.wc.risk.intro} have been studied in \cite{Ghossoubetal2023} for spectral risk measures and aggregation functions and in \cite{zbMATH07395113} for general tail risk measures.

Theorem \ref{thm.asymptotic.worst.case.risk} shows that, for a coherent risk measure $R\colon L^p(\Om)\to \R$ and $Z\in L^p(\Om)$, the ratio between the worst-case aggregate risk and the aggregate risk under pooling behaves asymptotically like
\[
\lim_{n\to \infty}\frac{R_n^{\rm WC}(Z)}{R\big(\sum_{i=1}^n\la_i X_i\big)}=\sup_{Y\sim Z}\frac{R(Y)}{\E[X_1]}<\infty.
\]
Again, the right-hand side only depends on the marginals and is completely independent of the joint distribution of the sequence $(X_n)_{n\in \N}$.\ If $R$ is also law-invariant, it follows that $\sup_{Y\sim Z}R(Y)=R(Z)$.\ We point out that the Expected Shortfall is a coherent law-invariant risk measure on $L^1(\Om)$ and therefore satisfies all these assumptions with $p=1$.

The rest of the paper is organized as follows.\ Section \ref{sec.uel} introduces the weighted Orlicz-uniform weak law of large numbers and proves the main characterization for risk measures on Orlicz spaces and $L^\infty(\Om)$.\ It also contains the unexpected-loss asymptotic for positively homogeneous risk measures.\ The asymptotic results on Choquet integrals for distorted probabilities with uniform moment bounds instead of uniform integrability are included in Section \ref{sec.insurance}. Section \ref{sec.risk.ratio} compares the aggregate risk of diversified benchmarks with co-monotonic alternatives, and Section \ref{sec.continuity} establishes the abstract continuity principles for monotone functionals on $L^0(\Om;B)$.

\section{Asymptotic Behaviour of Unexpected Losses}\label{sec.uel}

Throughout, let $(\Om,\Fc,\P)$ be a probability space and $\Phi\colon [0,\infty)\to [0,\infty)$ be a convex and continuous function that satisfies the $\Delta_2$-condition, i.e., there exists a constant $K>0$ such that
\[
0=\Phi(0)\leq \Phi(2t)\leq K\Phi(t)\quad \text{for all }t\geq 0,
\]
cf.\ \cite[Definition 1, p.\ 22]{RaoRen1991}.\ The $\Delta_2$-condition is a growth restriction for Young functions that ensures Orlicz spaces exhibit well-behaved structural and functional properties.\ We refer to \cite[Section 2.3]{RaoRen1991} for a detailed discussion of this growth condition and to \cite[Section 3.4]{RaoRen1991} for its role in the theory of Orlicz spaces.

 Let $B$ be a Banach lattice and $\Bc$ be its Borel $\si$-algebra.\ We say that a map $X\colon \Om\to B$ is a \textit{random variable} if $X$ is strongly measurable, i.e., $\Fc$-$\Bc$-measurable and there exists a $\P$-null set $N\in \Fc$ such that $X(\Om\setminus N):=\{X(\om)\colon \om\in \Om\setminus N\}$ is separable.\ Moreover, $L^\Phi(\Om;B)$ denotes the space of all random variables $X\colon \Om\to B$ with
\[
\E\big[\Phi\big(\|X\|\big)\big]<\infty,
\]
which, under the $\Delta_2$-condition, is equivalent to $\E\big[\Phi\big(k\|X\|\big)\big]<\infty$ for all $k>0$.

We point out that the case $\Phi\equiv0$ is not excluded.\ In this case, $L^\Phi(\Om;B)=L^0(\Om;B)$ and we equip $L^\Phi(\Om;B)$ with the topology induced by convergence in probability, see Section \ref{sec.continuity} below for the details.\ If $\Phi\not\equiv 0$, $\Phi(t)\to \infty$ as $t\to \infty$ and we equip $L^\Phi(\Om;B)$ with the usual Luxemburg norm
\[
\|X\|_\Phi:=\inf\Big\{k\in (0,\infty)\colon \E\Big[\Phi\Big(\tfrac{\|X\|}k\Big)\Big]\leq 1\Big\}\quad \text{for }X\in L^\Phi(\Om;B).
\]
Throughout, we say that a map $R\colon L^\Phi(\Om;B)\to B$ is a \textit{risk measure} if
\begin{enumerate}
\item[(i)] $R(X_1)\leq R(X_2)$ for all $X_1,X_2\in L^\Phi(\Om;B)$ with $X_1\leq X_2$ $\P$-a.s.,
\item[(ii)] $R(X+m)=R(X)+m$ for all $X\in L^\Phi(\Om;B)$ and $m\in B$.
\end{enumerate}
For a detailed discussion of risk measures in the case $B=\R$, we refer to \cite[Chapter 4]{FoellmerSchied2025}.\ In the scalar case, the Orlicz framework used here is closely related to the Orlicz-heart setting of \cite{CheriditoLi2008,CheriditoLi2009}, where monetary risk measures on Orlicz hearts are studied and dual representations are obtained in the coherent and convex case.\ Under the $\Delta_2$-condition imposed above, the Orlicz space $L^\Phi(\Om;B)$ coincides with the corresponding Orlicz heart, so the domain considered here is compatible with that setting. 

We consider a fixed sequence $(\lambda_n)_{n\in \N}\subset (0,\infty)$ of portfolio weights and define $\overline{\la}_n:=\frac1n\sum_{i=1}^n\la_i$.\ We say that a sequence of random variables $(X_n)_{n\in \N}\subset L^\Phi(\Om;B)$ satisfies the \textit{$\la$-weighted and $\Phi$-uniformly integrable weak law of large numbers} ($\la$-WLLN$_\Phi$) if there exists a sequence $(m_n)_{n\in \N}\subset B$ such that
\begin{equation}\label{eq.wlln}
\lim_{n\to \infty}\P\bigg(\bigg\|\frac1{n\overline{\la}_n}\sum_{i=1}^n \la_i(X_i-m_i)\bigg\|>\ep\bigg)= 0\quad \text{for all }\ep>0,
\end{equation} 
and
\begin{equation}\label{eq.unif.integrability}
\text{the sequence }\Bigg(\Phi\bigg(\bigg\|\frac1{n\overline\la_n}\sum_{i=1}^n \la_i(X_i-m_i)\bigg\|\bigg)\Bigg)_{n\in \N}\text{ is uniformly integrable.}
\end{equation}

Before stating the main result, we briefly discuss ($\la$-WLLN$_\Phi$) in the following remark.

\begin{remark}\label{rem.discussion.wlln}\
\begin{enumerate}
\item[a)] In the case $\Phi\equiv 0$, the uniform integrability condition \eqref{eq.unif.integrability} in ($\la$-WLLN$_\Phi$) is automatically satisfied, so that ($\la$-WLLN$_\Phi$) reduces to the classical $\la$-weighted weak law of large numbers \eqref{eq.wlln}.\ 
Observe that \eqref{eq.wlln} is the convergence in probability of the triangular weighted averages $\sum_{i=1}^n \frac{\la_i}{n\overline\la_n}(X_i-m_i)$. If $B=\R$ and $(X_n)_{n\in\N}$ is i.i.d.\ with $\E[|X_1|]<\infty$ and $m_i=\E[X_1]$, the Toeplitz condition
\begin{equation}\label{eq.condition.weights}
\max_{ i=1,\ldots, n}\frac{\lambda_i}{n\overline{\la}_n}\to 0\quad \text{as } n\to \infty,
\end{equation}
is equivalent to convergence in probability of the weighted average  to $\E[X_1]$, cf.\ \cite[Theorem 1]{Pruitt1966}.\ We also refer to \cite{Chou2023} for moment conditions ensuring convergence of suitably normalized sums without imposing an explicit dependence structure.\ For the classical strong law of large numbers with $\lambda_n=1$ for all $n\in \N$ in the case where $B$ is a separable Banach space, we refer to \cite[Corollary 7.10]{zbMATH00049190}.
\item[b)] In the case $\Phi\not\equiv 0$, a sequence of random variables $(X_n)_{n\in \N}\subset L^\Phi(\Om;B)$ satisfies ($\la$-WLLN$_\Phi$) with $(m_n)_{n\in \N}\subset B$ if and only if 
\[
\lim_{n\to \infty}\bigg\|\frac1{n\overline\la_n}\sum_{i=1}^n \la_i(X_i-m_i)\bigg\|_\Phi=0.
\]
Indeed, since $\Phi$ satisfies the $\Delta_2$-condition, a sequence $(Y_n)_{n\in \N}\subset L^\Phi(\Om;B)$ converges to the origin if and only if $\E[\Phi(\|Y_n\|)]\to 0$ as $n\to \infty$, which, by the Lebesgue-Vitali theorem \cite[Vol.\ I, Theorem 4.5.4, p.\ 268]{Boga2007} and the continuity of $\Phi$ at $0$, is equivalent to $Y_n\to 0$ in probability as $n\to \infty$ together with uniform integrability of the sequence $\big(\Phi(\|Y_n\|)\big)_{n\in \N}$.
\item[c)] In the case $\Phi(t)=t^p$ for $t\geq0$ with $p\in [1,\infty)$, a sequence $(X_n)_{n\in \N}\subset L^\Phi(\Om;B)$ satisfies \eqref{eq.unif.integrability} if and only if 
\begin{equation}\label{eq.delavalee}
\sup_{n\in \N}\E\bigg[\ph\bigg(\bigg\|\frac{1}{n\overline\la_n}\sum_{i=1}^n\la_i(X_i-m_i)\bigg\|\bigg)\bigg]<\infty
\end{equation}
for some function $\ph\colon [0,\infty)\to [0,\infty)$ with $\frac{\ph(t)}{t^p}\to \infty$ as $t\to \infty$, see e.g.\ \cite[Vol.\ I, Theorem 4.5.9]{Boga2007}.\ A common choice is $\ph(t)=t^q$ for $t\geq 0$ with exponent $q>p$.\ In this case, by the triangle inequality, \eqref{eq.delavalee} is satisfied if
\[
\sup_{n\in \N}\E\big[\|X_n-m_n\|^q\big]<\infty,
\]
which, in the case $q=2$, resembles uniformly bounded variances or standard deviations.\ Considering the Hilbert-space-valued case with $m_n=\E[X_n]$ for all $n\in \N$, pairwise uncorrelatedness of the sequence $(X_n)_{n\in \N}$ together with uniformly bounded variances and \eqref{eq.condition.weights} already implies the $\la$-weighted weak law of large numbers \eqref{eq.wlln} by Chebyshev's inequality.
\end{enumerate}
\end{remark}

\begin{theorem}\label{thm.main.orlicz}
  Let $B$ be a Banach lattice and $R\colon L^\Phi(\Om;B)\to B$ be a risk measure with $R(0)=0$.\ Then, the following statements are equivalent:
  \begin{enumerate}
    \item[(i)] For every $X\in L^\Phi(\Om;B_+)$, 
   \[    \lim_{\la\to 0} R(\la X)=0. \]
\item[(ii)] For every sequence of random variables $(X_n)_{n\in \N}\subset L^\Phi(\Om;B)$ that satisfies {\rm ($\la$-WLLN$_\Phi$)} with $(m_n)_{n\in \N}\subset B$,
\[
\lim_{n\to \infty} \bigg\|R\bigg(\frac1{n\overline\la_n}\sum_{i=1}^n \la_i X_i\bigg)-\frac1{n\overline\la_n}\sum_{i=1}^n \la_im_i\bigg\|=0.
\]
  \end{enumerate}
\end{theorem}

\begin{proof}
 Considering $X_n:=\frac{n\overline\la_n}{\la_n} Y_n-\sum_{i=1}^{n-1}\frac{\la_i}{\la_n}X_i$ for all $n\in \N$ with an arbitrary sequence $(Y_n)_{n\in \N}\subset L^\Phi(\Om;B)$ that converges to the origin in $L^\Phi(\Om;B)$, the statement in (ii) is equivalent to continuity at the origin of $R$. For $\Phi\equiv 0$, the equivalence of (i) and (ii) then follows from Theorem \ref{thm.cont.lattice.L0} choosing $L=B$ and $p$ as the norm on $B$.\ Now, let $\Phi\not\equiv 0$. Then, (WLLN$_\Phi$) is equivalent to convergence of the sequence $\big(\frac1{n\overline\la_n}\sum_{i=1}^n\la_i(X_i-m_i)\big)_{n\in \N}$ to the origin in $L^\Phi(\Om;B)$ by Remark \ref{rem.discussion.wlln} b).\ Hence, the equivalence of (i) and (ii) again follows from Theorem \ref{thm.cont.lattice.L0} in the case $\#\Om=1$, see Remark \ref{rem.omegaone}, choosing $L=B$, $p$ as the norm $\|\cdot\|$ on $B$, and $L^\Phi(\Om;B)$ as the underlying Banach lattice in Theorem \ref{thm.cont.lattice.L0}. 
\end{proof}

\begin{remark}\label{rem.uel}
We point out that any risk measure $R\colon L^\Phi(\Om;B)\to B$ that is \textit{positively homogeneous}, i.e.,
\[
R(\la X)=\la R(X) \quad \text{for all }\la\in (0,\infty) \text{ and }X\in L^\Phi(\Om;B),
\]
satisfies $R(0)=0$ and condition (i) in the previous theorem.\ In this case, for every sequence of random variables $(X_n)_{n\in \N}\subset L^\Phi(\Om;B)$ that satisfies {\rm ($\la$-WLLN$_\Phi$)} with $(m_n)_{n\in \N}\subset B$,
\begin{equation}\label{eq.uel}
\lim_{n\to \infty} \frac1{n\overline\la_n}\bigg\|R\bigg(\sum_{i=1}^n \la_i X_i\bigg)-\sum_{i=1}^n \la_im_i\bigg\|=0.
\end{equation}
Hence, the unexpected losses
\[
{\rm UEL}_n:=R\bigg(\sum_{i=1}^n \lambda_i X_i\bigg)-\sum_{i=1}^n \lambda_im_i
\]
satisfy $\|{\rm UEL}_n\|=o\big(n\overline\lambda_n\big)$ as $n\to\infty$, i.e., after normalizing by the total portfolio weight, the risk-measure correction above the deterministic benchmark vanishes asymptotically.\ We therefore see that pooling effects automatically carry over from expected losses to unexpected losses.\ If $m_n=m$ for all $n\in \N$, we further have that
\[
\lim_{n\to \infty} \frac1{n\overline\la_n}R\bigg(\sum_{i=1}^n \la_i X_i\bigg)=m.
\]
\end{remark}

We conclude this section with a characterization for risk measures defined on the space $L^\infty(\Om)$ of all ($\P$-equivalence classes of) $\P$-essentially bounded measurable functions $\Om\to \R$, which builds on the classical notions of continuity from above and below, \cite[Lemma 4.21 and Theorem 4.22]{FoellmerSchied2025}.

For a sequence of random variables $(X_n)_{n\in \N}\subset L^\infty(\Om)$, we write $X_n\searrow 0$ or $X_n\nearrow 0$ $\P$-a.s.\ as $n\to \infty$ if $X_{n+1}\leq X_n$ or $X_{n+1}\geq X_n$ $\P$-a.s.\ for all $n\in \N$ and $X_n\to 0$ $\P$-a.s.\ as $n\to\infty$, respectively.\ We say that a risk measure $R\colon L^\infty(\Om)\to \R$ is \textit{continuous from above or below at zero} if $\lim_{n\to \infty}R(X_n)=R(0)$ for any sequence of random variables $(X_n)_{n\in \N}\subset L^\infty(\Om)$ with $X_n\searrow0$ or $X_n\nearrow 0$ $\P$-a.s.\ as $n\to \infty$, respectively.

We say that a sequence of random variables $(X_n)_{n\in \N}\subset L^\infty(\Om)$ satisfies ($\la$-WLLN$_\infty$) if there exists a sequence $(m_n)_{n\in \N}\subset \R$ such that \eqref{eq.wlln} is satisfied with $B=\R$ and
\begin{equation}\label{eq.unif.bddedness}
\text{the sequence }\bigg(\frac1{n\overline\la_n}\sum_{i=1}^n \la_i(X_i-m_i)\bigg)_{n\in \N}\text{ is $\P$-a.s.\ uniformly bounded.}
\end{equation}

\begin{theorem}\label{thm.main.infty}
  Let $R\colon L^\infty(\Om)\to \R$ be a risk measure with $R(0)=0$.\ Then, the following statements are equivalent:
  \begin{enumerate}
    \item[(i)] $R$ is continuous from above and below at zero.
\item[(ii)] For every sequence of random variables $(X_n)_{n\in \N}\subset L^\infty(\Om)$ that satisfies {\rm ($\la$-WLLN$_\infty$)} with $(m_n)_{n\in \N}\subset \R$,
\[
\lim_{n\to \infty} \bigg|R\bigg(\frac1{n\overline \la_n}\sum_{i=1}^n \la_i X_i\bigg)-\frac1{n\overline \la_n}\sum_{i=1}^n \la_i m_i\bigg|=0.
\]
  \end{enumerate}
\end{theorem}

\begin{proof}
 Again, considering $X_n:=\frac{n\overline\la_n}{\la_n}Y_n-\sum_{i=1}^{n-1}\frac{\la_i}{\la_n}X_i$ for all $n\in \N$ with an arbitrary sequence $(Y_n)_{n\in \N}\subset L^\infty(\Om)$ that is $\P$-a.s.\ uniformly bounded and converges to zero in probability, (ii) is equivalent to $\lim_{n\to \infty}R(Y_n)= 0$ for all sequences $(Y_n)_{n\in \N}\subset L^\infty(\Om)$ that are $\P$-a.s.\ uniformly bounded and converge to zero in probability.\ The latter clearly implies (i) since every nondecreasing or nonincreasing sequence that converges $\P$-a.s.\ to zero is $\P$-a.s.\ uniformly bounded and converges to zero in probability.\ Now, assume that $R$ is continuous from above and below at zero and let $(Y_n)_{n\in \N}\subset L^\infty(\Om)$ be $\P$-a.s.\ uniformly bounded and converge to zero in probability.\ By the Borel-Cantelli lemma, every subsequence of $(Y_n)_{n\in \N}$ has a further subsequence $(X_n)_{n\in \N}$ that converges to zero $\P$-a.s.\ Let
 \[
 X_n^1:=\inf_{k\geq n} X_k\quad \text{ and } X_n^2:=\sup_{k\geq n} X_k\quad \text{for all }n\in \N.
 \]
 Then, $X_n^1\nearrow 0$ $\P$-a.s.\ as $n\to \infty$ and $X_n^2\searrow 0$ $\P$-a.s.\ as $n\to \infty$.\ Since $X_n^1\leq X_n\leq X_n^2$ for all $n\in \N$, it follows that
 \[
  0=\lim_{n\to \infty}R\big(X_n^1\big)\leq  \liminf_{n\to \infty}R(X_n)\leq \limsup_{n\to \infty}R(X_n)\leq \lim_{n\to \infty}R\big(X_n^2\big)= 0.
 \]
 We have thus shown that $\lim_{n\to \infty} R(X_n)=0$, which implies that $\lim_{n\to \infty} R(Y_n)=0$.
\end{proof}

\section{Asymptotic Behaviour of Insurance Premia}\label{sec.insurance}

In this section, we focus on pooling effects for a general class of insurance premia.\ Actuarial applications typically focus on nonnegative random variables, since the risks to be pooled are individual losses.\ As before, let  $(\lambda_n)_{n\in \N}\subset (0,\infty)$ be a fixed sequence of portfolio weights.\ We adopt the axiomatic characterization of \cite{WYP1997} and consider the set $\Xc$ of all nonnegative real-valued random variables on an atomless probability space $(\Om,\Fc,\P)$.\ We then consider an insurance premium $H\colon \Xc\to [0,\infty]$ of the form
\begin{equation}\label{eq.inspremia}
H(X)=C\int_0^\infty h\big(\P(X>x)\big)\, \d x
\end{equation}
with a constant $C>0$ and a distortion function $h\colon [0,1]\to [0,1]$, which is assumed to be nondecreasing with $\lim_{u\downarrow 0}h(u)=0$ and $\lim_{u\uparrow 1}h(u)=1$.\ We point out that $H(0)=0$, $H(1)=C>0$, and 
\begin{equation}\label{eq.premium.pos.hom}
H(\al X)=\al H(X) \quad\text{for all }\al>0\text{ and }X\in \Xc,
\end{equation}
cf.\ \cite{WYP1997}, where an axiomatic foundation for insurance premia of the form \eqref{eq.inspremia} is given.

 More generally, $H$ belongs to the class of Choquet integrals, see \cite{WYP1997}, where also an axiomatization for this broader class of insurance premium principles based on Greco's representation theorem \cite{zbMATH03830277} is provided.\ If $c\colon \Fc\to [0,\infty)$ is a Choquet capacity, i.e., a monotone set function with $c(\Om)>0$ and $c(A)=0$ for all $A\in \Fc$ with $\P(A)=0$, the Choquet integral of $X\in L^\infty(\Om)$ is given by
 \[
 \int X\, \d c:=\int_{-\infty}^0 \big(c(X>x)-c(\Om)\big)\, \d x+\int_0^{\infty} c(X>x)\, \d x.
 \]
 In this case, by \cite[Proposition 4.90]{FoellmerSchied2025},
 \[
 R(X):=\frac1{c(\Om)} \int X {\rm d} c
 \]
 defines a positively homogeneous risk measure on $ L^\infty(\Om)$. If $$\lim_{n\to \infty}c(A_n)=0\quad \text{and} \quad \lim_{n\to \infty} c(B_n)=c(\Om)$$ for all sequences $(A_n)_{n\in \N}\subset \Fc$ and $(B_n)_{n\in \N}\subset \Fc$ that are nonincreasing and nondecreasing with $\bigcap_{n\in \N}A_n=\emptyset$ and $\bigcup_{n\in \N}B_n=\Om$, respectively, $R$ is continuous from above and below at zero by dominated convergence, so that, by Theorem \ref{thm.main.infty},
\[
\lim_{n\to \infty} \frac{1}{n\overline \la_n}\bigg|R\bigg(\sum_{i=1}^n \la_i X_i\bigg)-\sum_{i=1}^n \la_i  m_i\bigg|=0
\]
 for every sequence of random variables $(X_n)_{n\in \N}\subset L^\infty(\Om)$ that satisfies ($\la$-WLLN$_\infty$) with $(m_n)_{n\in \N}\subset \R$.

However, the Choquet integral and the distorted premium principle in \eqref{eq.inspremia} do not necessarily take finite values on $L_+^p(\Om):=L^p(\Om;\R)\cap \Xc$.\ Throughout this section, we therefore work under the assumption that
\begin{equation}\label{eq.ass.h}
 \int_1^\infty h\big(\tfrac1{x^p}\big)\, \d x<\infty\quad \text{for some }p\in [1,\infty),
\end{equation}
and set $h(t):=1$ for $t>1$.\ Observe that, for a uniformly distributed random variable $U$,
\[
H\big(U^{-1/p}\big)=C\int_1^\infty h\Big(\P\big(U^{-1/p}>x\big)\Big)\, \d x=C\int_1^\infty h\Big(\P\big(U<\tfrac1{x^p}\big)\Big)\, \d x=C\int_1^\infty h\big(\tfrac1{x^p}\big)\, \d x.
\]
Hence, the condition $\int_1^\infty h\big(\frac1{x^p}\big)\, \d x<\infty$ entails that $H\big(U^{-1/p}\big)<\infty$ for uniformly distributed random variables $U$.\
Moreover, by Markov's inequality, for all $X\in \Xc$ with $\E[X^p]<\infty$,
\begin{align*}
 H(X)&=C\int_0^\infty h\big(\P(X>x)\big)\,\d x\leq C\E[X^p]^{\frac1p}h(1)+C\int_{\E[X^p]^{\frac1p}}^\infty h\Big(\tfrac{\E[X^p]}{x^p}\Big)\,\d x\\
 &= C\E[X^p]^{\frac1p}\bigg(h(1)+\int_1^\infty h\big(\tfrac{1}{x^p}\big)\,\d x\bigg)<\infty.
\end{align*}
Hence, the integrability condition \eqref{eq.ass.h} for the distortion function $h$ yields that the premium principle $H$ takes finite values on $L^p_+(\Om)$.\ Moreover, the integrability condition \eqref{eq.ass.h} allows us to replace the uniform integrability condition in \eqref{eq.unif.integrability} with $\Phi(x)=x^p$ by a simple moment constraint of order $p$, whereas uniform integrability is typically verified by imposing a stronger condition of uniformly bounded moments of order $q>p$.\ We have the following result.

\begin{theorem}\label{thm.main.insurance}
Let $p\in [1,\infty)$ and $H\colon L^p_+(\Om)\to \R$ be given by \eqref{eq.inspremia} with a constant $C>0$ and a nondecreasing distortion function $h\colon [0,1]\to [0,1]$ that satisfies $\lim_{u\downarrow 0}h(u)=0$, $\lim_{u\uparrow 1}h(u)=1$, and $\int_1^\infty h\big(\frac1{x^p}\big)\,\d x<\infty$.\ Moreover, let $(X_n)_{n\in \N}\subset \Xc$ with $\E[X_n]=:m$ for all $n\in \N$, $\sup_{n\in \N}\E[X_n^p]<\infty$, and
\[
\lim_{n\to \infty}\P\bigg(\bigg|\frac1{n\overline\la_n}\sum_{i=1}^n \la_i(X_i-m)\bigg|>\ep\bigg)= 0\quad \text{for all }\ep>0.
\]
Then,
 \[
\lim_{n\to \infty} \bigg|\frac{1}{n\overline\la_n}H\bigg(\sum_{i=1}^n\la_iX_i\bigg)-Cm\bigg|=0.
 \]
\end{theorem}

\begin{proof}
If $m=0$, it follows that $X_n=0$ $\P$-a.s.\ for all $n\in \N$ and the statement is trivial.\ Therefore, assume that $m>0$, and let $Y_n:=\frac1{n\overline\la_n}\sum_{i=1}^n \la_iX_i$ for all $n\in \N$.\ Then, by assumption, $Y_n\to m$ in probability as $n\to \infty$ and, using Jensen's inequality,
\[
0<m\leq M:=\sup_{n\in \N}\E[Y_n^p]^{\frac1p}\leq \sup_{n\in \N}\E[X_n^p]^{\frac1p}<\infty.
\]
Let $\ep>0$.\ Since $\int_1^\infty h\big(\tfrac1{x^p}\big)\, \d x<\infty$, there exists some $r\geq 1$ such that
\[
\int_r^\infty h\big(\tfrac1{x^p}\big)\, \d x\leq \frac{\ep}{CM}.
\]
Hence, by Markov's inequality, the dominated convergence theorem, and the continuity of $h$ at $0$ and $1$,
\begin{align*}
\big|H(Y_n)-Cm\big|&\leq C\bigg|\int_0^{rM} h\big(\P(Y_n>x)\big)\,\d x -m\bigg| + C\int_{rM}^\infty h\big(\tfrac{M^p}{x^p}\big)\,\d x\\
& \leq C\bigg|\int_0^{rM} h\big(\P(Y_n>x)\big)\,\d x -m\bigg|+CM\int_{r}^\infty h\big(\tfrac{1}{x^p}\big)\,\d x\\
&\leq C\bigg|\int_0^{rM} h\big(\P(Y_n>x)\big)\,\d x -m\bigg|+\ep\to \ep\quad \text{as }n\to \infty.
\end{align*}
Letting $\ep\downarrow 0$, the claim follows.
\end{proof}

\section{Risk Ratios for Aggregate Risk of Co-Monotonic Alternatives}\label{sec.risk.ratio}

The previous sections identify conditions under which diversification makes the nonlinear aggregate risk asymptotically comparable to the portfolio mean.\ This section analyzes how misspecification of the dependence structure affects aggregate risk.\ In credit and insurance portfolios, marginal loss distributions may be estimated from ratings, underwriting classes, or historical severity data, while joint dependence is much harder to infer. The recent paper \cite{DeVecchi2026} shows that even a small amount of positive dependence, if it cannot be excluded, may allow tail losses to behave as if they were perfectly dependent beyond a certain threshold, leading to a potential risk underestimation if pooling behaviour is assumed for the losses. In this section, we abstract this phenomenon by comparing an insurable benchmark with a co-monotonic alternative.\ We quantify the limiting factor by which the aggregate risk is potentially underestimated if the co-monotonic alternative cannot be excluded.

We say that a risk measure $R\colon L^p(\Om)\to \R$ is \textit{co-monotonic additive} if
\[
R\bigg(\sum_{i=1}^n Y_i\bigg)=\sum_{i=1}^nR(Y_i)
\]
for all random variables $Y_1,\ldots, Y_n\in L^p(\Om)$ that are \textit{co-monotonic}, cf.\ \cite[Definition 4.88 and Lemma 4.95]{FoellmerSchied2025} for a definition and equivalent descriptions of co-monotonicity. We point out that every co-monotonically additive risk measure is positively homogeneous, see, e.g., \cite[Lemma 4.89]{FoellmerSchied2025} for the case of bounded random variables, so that every co-monotonically additive risk measure on  $L^p(\Om)$ automatically satisfies condition (i) in Theorem \ref{thm.main.orlicz}, see also Remark \ref{rem.uel}.

Throughout this section, we work under the following assumption, where $\Xc$ denotes again the space of all nonnegative random variables $\Om\to \R$.

\begin{assumption}\label{ass.application}
 Let $1\leq p\leq q<\infty$ and $(\la_n)_{n\in \N}\subset (0,\infty)$ be a sequence of portfolio weights.\ We consider a sequence $(X_n)_{n\in \N}\subset L^p(\Om)$ with 
 \begin{equation}\label{eq.ass.risk.ratio}
 \sup_{n\in \N} \E\big[|X_n|^q\big]<\infty\quad \text{and}\quad \E[X_n]=\E[X_1]>0\quad \text{for all }n\in \N,
 \end{equation}
 that satisfies the $\la$-weighted weak law of large numbers \eqref{eq.wlln} with $m_n=\E[X_1]$ for all $n\in \N$ and a second sequence $(Z_n)_{n\in \N}\subset L^p(\Om)$ such that $Z_1,\ldots, Z_n$ are co-monotonic for each $n\in \N$.
\end{assumption}

Assumption \ref{ass.application} specifies a portfolio of insurable claims or diversified credits via the sequence of random variables $(X_n)_{n\in \N}$ that satisfies the weak law of large numbers with a moment constraint.\ In the standard case $p=2$, \eqref{eq.ass.risk.ratio} together with \eqref{eq.condition.weights} implies the $\la$-weighted weak law of large numbers \eqref{eq.wlln} with $m_n=\E[X_1]$ for all $n\in \N$ if the sequence $(X_n)_{n\in \N}$ is pairwise uncorrelated, see Remark \ref{rem.discussion.wlln} c).\ On the other hand, the sequence of random variables $(Z_n)_{n\in \N}$ describes an adversarial aggregation scenario without benefits from diversification.\ The ratio in the following theorem can therefore be interpreted as a comparison between the aggregate risk of a portfolio that allows for pooling and a co-monotonic alternative that cannot be ruled out due to dependence uncertainty.

\begin{theorem}\label{thm.asymptotic.risk.ratio}
Let Assumption \ref{ass.application} be satisfied with $p<q$ and $R\colon L^p(\Om)\to \R$ be a co-monotonic additive risk measure.\ If $R(Z_n)=R(Z_1)$ for all $n\in \N$, then
\begin{equation}\label{eq.asymptotic.risk.ratio}
\lim_{n\to \infty}\frac{R\big(\sum_{i=1}^n \la_iZ_i\big)}{R\big(\sum_{i=1}^n \la_iX_i\big)}=\frac{R(Z_1)}{\E[X_1]}.
\end{equation}
\end{theorem}

\begin{proof}
By assumption and Theorem \ref{thm.main.orlicz} together with Remark \ref{rem.discussion.wlln} c),
\[
\lim_{n\to \infty}\frac{R\big(\sum_{i=1}^n \la_iZ_i\big)}{R\big(\sum_{i=1}^n \la_iX_i\big)}=\lim_{n\to \infty}\frac{R(Z_1)}{R\big(\tfrac{1}{n\overline\la_n}\sum_{i=1}^n \la_iX_i\big)}=\frac{R(Z_1)}{\E[X_1]}.
\]
\end{proof}

\begin{remark}\label{rem.risk.ratio}
We point out that the right-hand side in \eqref{eq.asymptotic.risk.ratio} is completely specified by the individual risk and distribution of $Z_1$ and $X_1$, respectively.\ Consider, for example, the case where $X_n \sim Z_n\sim B(1,p)$, i.e., $\P(Z_n=1)=1-\P(Z_n=0)=p\in (0,1)$ for all $n\in \N$.\ If $R={\rm VaR}^\alpha$ is the Value at Risk with level $\alpha\in (1-p,1)$, we have 
\[
\frac{{\rm VaR}^\alpha(Z_1)}{\E[X_1]}=\frac{1}{p}.
\]
For the expected shortfall ${\rm ES}^\alpha$ at level $\alpha\in (0,1)$, we get
\[
\frac{{\rm ES}^\alpha(Z_1)}{\E[X_1]}=\frac{1}{\max\{p,1-\alpha\}}.
\]
Now, if $X_n \sim Z_n\sim \exp (\la)$ are exponentially distributed with rate $\la>0$ for all $n\in \N$, we get
\[
\frac{{\rm VaR}^\alpha(Z_1)}{\E[X_1]}=-\log(1-\alpha)\quad \text{and}\quad \frac{{\rm ES}^\alpha(Z_1)}{\E[X_1]}=1-\log(1-\alpha).
\]
\end{remark}
By \cite[Theorem 4.94]{FoellmerSchied2025}, Choquet integrals are co-monotonic additive, leading to the following analogous result for insurance premia in the form of distorted probabilities.\  

\begin{theorem}\label{thm.asymptotic.ins.ratio}
Let Assumption \ref{ass.application} be satisfied with $p=q$ and $H\colon L^p_+(\Om)\to \R$ be given by \eqref{eq.inspremia} with $C>0$ and a nondecreasing distortion function $h\colon [0,1]\to [0,1]$ that satisfies $\lim_{u\downarrow 0}h(u)=0$, $\lim_{u\uparrow 1}h(u)=1$, and $\int_1^\infty h\big(\frac1{x^p}\big)\,\d x<\infty$.\ If, additionally, $X_n,Z_n\in \Xc$ and $H(Z_n)=H(Z_1)$ for all $n\in \N$, then
\begin{equation}\label{eq.asymptotic.ins.ratio}
\lim_{n\to \infty}\frac{H\big(\sum_{i=1}^n \la_iZ_i\big)}{H\big(\sum_{i=1}^n \la_iX_i\big)}=\frac1{\E[X_1]}\int_0^\infty h\big(\P(Z_1> x)\big)\, \d x.
\end{equation}
\end{theorem}

\begin{proof}
By assumption and Theorem \ref{thm.main.insurance},
\[
\lim_{n\to \infty}\frac{H\big(\sum_{i=1}^n \la_iZ_i\big)}{H\big(\sum_{i=1}^n \la_iX_i\big)}=\lim_{n\to \infty}\frac{H(Z_1)}{H\big(\tfrac{1}{n\overline\la_n}\sum_{i=1}^n \la_iX_i\big)}=\frac1{\E[X_1]}\int_0^\infty h\big(\P(Z_1> x)\big)\, \d x.
\]
\end{proof}

We conclude this section with the special case where $R\colon L^p(\Om)\to \R$ is a coherent risk measure.\ For a given random variable $Z\in L^p(\Om)$, a risk measure $R\colon L^p(\Om)\to \R$ and $n\in \N$, we consider the worst-case aggregate risk
\[
R_n^{\rm WC}(Z):=\sup_{Y_1\sim Z,\ldots, Y_n\sim Z} R\bigg(\sum_{i=1}^n\la_i Y_i\bigg),
\]
where the supremum is taken over all random variables $Y_1,\ldots, Y_n\in L^p(\Om)$ that follow the same distribution as $Z$.\

We have the following theorem that describes the asymptotic ratio between the worst-case aggregate risk and the aggregate risk under the assumption of pooling and uniformly bounded moments of order $q>p$.

\begin{theorem}\label{thm.asymptotic.worst.case.risk}
Let Assumption \ref{ass.application} be satisfied with $p<q$ and $R\colon L^p(\Om)\to \R$ be a coherent risk measure.\ Then, for all $Z\in L^p(\Om)$,
\[
 \lim_{n\to \infty}\frac{R_n^{\rm WC}(Z)}{R\big(\sum_{i=1}^n\la_i X_i\big)}=\sup_{Y\sim Z}\frac{R(Y)}{\E[X_1]}<\infty.
\]
\end{theorem}

\begin{proof}
    Since $R$ is coherent, by Corollary \ref{cor.cont.poshom}, there exists a constant $C\geq 0$ such that
    \[
      \sup_{Y\sim Z}|R(Y)|\leq \sup_{Y\sim Z}C\|Y\|_p=C\|Z\|_p<\infty.
    \]
    Moreover, for all $n\in \N$
    \[
     \sup_{Y\sim Z}\sum_{i=1}^n\la_iR(Y) \leq R_n^{\rm WC}(Z)\leq \sum_{i=1}^n\la_i\sup_{Y\sim Z}R(Y)=\sup_{Y\sim Z}\sum_{i=1}^n\la_iR(Y).
    \]
    By positive homogeneity of $R$ and Theorem \ref{thm.main.orlicz} with Remark \ref{rem.discussion.wlln} c), we thus obtain
    \[
    \lim_{n\to \infty}\frac{R_n^{\rm WC}(Z)}{R\big(\sum_{i=1}^n\la_i X_i\big)} =\lim_{n\to \infty}\sup_{Y\sim Z}\frac{R(Y)}{R\big(\tfrac{1}{n\overline\la_n}\sum_{i=1}^n \la_iX_i\big)}=\sup_{Y\sim Z}\frac{R(Y)}{\E[X_1]}<\infty.
    \]
\end{proof}

\section{Continuity at the Origin and Monotonic Functionals on $L^0$}\label{sec.continuity}

Throughout this section, let $(\Om,\Fc,\P)$ be a probability space. For a metric space $(S,d)$, endowed with the Borel $\sigma$-algebra $\Bc$, we denote by $L^0(\Om;S)$ the set of all ($\P$-equivalence classes of) strongly measurable functions taking values in $S$ $\P$-a.s.\ Recall that a function $X\colon \Omega \to S$ is strongly measurable if and only if it is $\Fc$-$\Bc$-measurable and there is a $\P$-null set $N\in \Fc$ such that $X(\Om\setminus N)$ is separable.\ We endow $L^0(\Om;S)$ with the topology of convergence in probability, i.e., the topology induced by the metric
\[
 d_\P(X,Y):=\E[d(X,Y)\wedge 1]\quad \text{for }X,Y\in L^0(\Om;S).
\]

\begin{remark}\label{rem.omegaone}
 Considering the case where $\#\Om=1$, i.e., $\Om$ is a singleton, with $\Fc$ being the power set, all results obtained in this section concerning continuity properties w.r.t.\ convergence in probability on $L^0(\Om;S)$ can be transferred to continuity properties on $S$, where $S=C$ will be a closed convex cone in a Banach space $B$ (including $C=B$).
\end{remark}

\begin{proposition}\label{prop.cont.cone.L0}
Let $B$ be a Banach space, $C\subseteq B$ be a closed convex cone, $L$ be a topological space, and $f\colon L^0(\Om;C)\to L$.\ Assume that, for every sequence $(X_n)_{n\in \N}\subset L^0(\Om;C)$,
\[
f(X_n)\to f(0) \quad \text{as }n\to \infty
\]
whenever there exist $X^*\in L^0(\Om;C)$ and a null sequence $(\la_n)_{n\in \N}\subset (0,\infty)$ such that $\lambda_n X^*-X_n\in C$ $\P$-a.s.\ for all $n\in \N$.\ 
Then, $f$ is continuous at the origin.
\end{proposition}

\begin{proof}
 Let $(Y_n)_{n\in \N}\subset L^0(\Om;C)$ be a sequence that converges to zero in probability and $(\la_n)_{n\in \N}\subset (0,\infty)$ be an arbitrary null sequence.\ Then, every subsequence of $(Y_n)_{n\in \N}$ has a further subsequence $(X_n)_{n\in \N}$ with $$\P\big( \|X_n\|> \lambda_n2^{-n}\big)\leq 2^{-n}\quad\text{for all }n\in \N.$$ Since $B$ is a Banach space and $C\subseteq B$ is a closed convex cone, by the Borel-Cantelli lemma, $X^*:=\sum_{k\in \N} \la_k^{-1}X_k\in C$ $\P$-a.s.\ and $$\la_n X^*-X_n=\sum_{\substack{k\in \N\\ k\neq n}} \tfrac{\la_n}{\la_k}X_k\in C\quad \P\text{-a.s.\ for all }n\in \N,$$
 so that $f(X_n)\to f(0)$ as $n\to \infty$.\ We have therefore shown that $f(Y_n)\to f(0)$ as $n\to \infty$ for every sequence $(Y_n)_{n\in \N}\subset L^0(\Om;C)$ that converges to zero in probability.\
\end{proof}

We now apply the previous proposition to monotonic functionals. To that end, we specialize to Banach lattices $B$ and the closed convex cone $$B_+:=\{x\in B\colon x\geq0\}.$$ 
Moreover, we consider a partially ordered set $(L,\leq)$.\ The order on $L$ naturally induces a topology, which is referred to as the interval topology.\ For $a,b \in L$, one considers sets of the form
\[
(-\infty, a] := \{ x \in L \colon  x\leq a \}\quad \text{and}\quad [b,\infty) := \{ x \in L \colon b \leq x \},
\]
and defines the \textit{interval topology}, cf.\ \cite[Section IV.8, p.\ 59--62]{Birkhoff1948} and \cite{Frink1942}, as the topology generated by
\[
 \big\{L\setminus (-\infty,a]\colon a\in L\big\}\cup  \big\{L\setminus [b,\infty)\colon b\in L\big\}.
\]
Moreover, we consider the \textit{lsc-topology} and the \textit{usc-topology}, which are the topologies generated by
\[
 \big\{L\setminus (-\infty,a]\colon a\in L\big\}\quad\text{and}\quad  \big\{L\setminus [b,\infty)\colon b\in L\big\},
\]
respectively.\ We say that a map $f\colon L^0(\Om;B)\to L$ is {lower} or {upper semicontinuous} if $f$ is continuous w.r.t.\ the lsc- or usc-topology on $L$, respectively.\ Moreover, we say that $f\colon L^0(\Om;B)\to L$ is continuous if it is continuous w.r.t.\ the interval topology on $L$ or, equivalently, if it is lower and upper semicontinuous.

\begin{theorem}\label{thm.cont.monotone.L0}
 Let $B$ be a Banach lattice, $(L,\leq)$ be a partially ordered set, endowed with the interval topology, $R\colon L^0(\Om;B)\to L$ be monotonic, and $X_0\in L^0(\Om;B)$.\ Then,
\begin{enumerate}
 \item[a)] $R$ is upper semicontinuous at $X_0$ w.r.t.\ convergence in probability if and only if
 \begin{equation}\label{eq.usc.monotone.L0}
 R(X_0)=\inf_{\la>0}R(X_0+\la X)\quad \text{for all }X\in L^0(\Om;B_+),
 \end{equation}
 \item[b)] $R$ is lower semicontinuous at $X_0$ w.r.t.\ convergence in probability if and only if
 \begin{equation}\label{eq.lsc.monotone.L0}
 R(X_0)=\sup_{\la>0}R(X_0-\la X)\quad \text{for all }X\in L^0(\Om;B_+),
 \end{equation}
  \item[c)] $R$ is continuous at $X_0$ w.r.t.\ convergence in probability if and only if
 \begin{equation}\label{eq.cont.monotone.L0}
 \sup_{\la>0}R(X_0-\la X)=R(X_0)=\inf_{\la>0}R(X_0+\la X)\quad \text{for all }X\in L^0(\Om;B_+).
 \end{equation}
\end{enumerate}
\end{theorem}

\begin{proof}
First, assume that \eqref{eq.usc.monotone.L0} is satisfied.\
W.l.o.g.,\ let $X_0=0$ or else consider $\overline R(X):=R(X_0+X)$ for all $X\in L^0(\Om;B)$ instead of $R$.\ Since $B$ is a Banach lattice, it follows that $B_+$ is a closed convex cone.\ Let $(X_n)_{n\in \N}\subset L^0(\Om;B_+)$ with $X_n\leq \lambda_n X^*$ for some $X^*\in L^0(\Om;B_+)$ and some null sequence $(\la_n)_{n\in \N}\subset (0,\infty)$ and $b\in L$ with $R(0)\notin [b,\infty)$.\ Then, there exists some $\la_0>0$ such that $R(\la X^*)\notin [b,\infty)$ for all $\la\in (0,\la_0)$.\ Otherwise, $b$ would be a lower bound of the set $\{R(\la X^*)\colon \la>0\}$ and therefore $b\leq R(0)$ by \eqref{eq.usc.monotone.L0}, which would be a contradiction. Since $R(X_n)\leq R(\la_nX^*)$ for all $n\in \N$, it follows that $R(X_n)\notin [b,\infty)$ for all $n\in \N$ with $\la_n<\la_0$, so that $R(X_n)\to R(0)$ as $n\to \infty$ in the usc-topology on $L$. By Proposition \ref{prop.cont.cone.L0}, for every sequence $(Y_n)_{n\in \N}\subset L^0(\Om;B)$ with $Y_n\to 0$ in probability as $n\to \infty$ and every $b\in L$, there exists $n_0\in \N$ such that $$R(Y_n)\leq R(Y_n^+)\notin [b,\infty)\quad \text{for all }n\in \N\text{ with }n\geq n_0,$$ so that $R$ is upper semicontinuous w.r.t.\ convergence in probability at the origin.

Now, assume that $R$ is upper semicontinuous w.r.t.\ convergence in probability at $X_0$, and let $X\in L^0(\Om;B_+)$.\ Let $b\in L$ be a lower bound for $\{R(X_0+\la X)\colon \la>0\}$, i.e., $R(X_0+\la X)\in [b,\infty)$ for all $\la>0$. Since $R$ is upper semicontinuous at $X_0$ w.r.t.\ convergence in probability, this implies that $R(X_0)\in [b,\infty)$, so that $R(X_0)\geq b$. By monotonicity of the map $R$, $R(X_0)$ itself is a lower bound for the set $\{R(X_0+\la X)\colon \la>0\}$, so that \eqref{eq.usc.monotone.L0} holds. We have therefore proved part a).

Part b) now follows from part a) by switching the order on $L$ and part c) is a direct consequence of a) and b).
\end{proof}

We now focus on the case where the partially ordered set $(L,\leq)$ is a vector lattice.\ We refer to \cite[Section V.7]{ScWo1999} for an introduction to topological vector lattices.\ For $u\in L$, we use the notation $$u_+:=u\vee 0\quad \text{and}\quad |u|:=u_++(-u)_+.$$ A map $p\colon L\to [0,\infty)$ is called a \textit{lattice seminorm} if it is a seminorm that is compatible with the lattice structure in the sense that
\[
p(u) \leq p(v)\quad \text{for all }u,v\in L\text{ with }|u| \leq |v|.
\]
Observe that, by definition,
\[
 p(u) \leq p(v)\quad \text{for all }u,v\in L\text{ with }0\leq u\leq v
\]
 for every lattice seminorm $p\colon L\to [0,\infty)$.

\begin{theorem}\label{thm.cont.lattice.L0}
 Let $B$ be a Banach lattice, $L$ be a vector lattice, $R\colon L^0(\Om;B)\to L$ be monotonic, $p\colon L\to [0,\infty)$ be a lattice seminorm and $X_0\in L^0(\Om;B)$.\
 Then, the map $$L^0(\Om;B)\to [0,\infty),\quad  X\mapsto p\big(R(X)-R(X_0)\big)$$ is continuous at $X_0$ if and only if
 \begin{equation}\label{eq.cont.lattice.monotone}
 \lim_{\la\to 0}p\big(R(X_0+\la X)-R(X_0)\big)=0\quad \text{for all }X\in L^0(\Om;B_+).
 \end{equation}
\end{theorem}

\begin{proof}
 Clearly, continuity of the map $$L^0(\Om;B)\to [0,\infty),\quad  X\mapsto p\big(R(X)-R(X_0)\big)$$ in $X_0$ implies \eqref{eq.cont.lattice.monotone}.\ On the other hand,
    \[
 p\Big(\big(R(X_0+X)-R(X_0)\big)_+\Big)\leq p\big(R(X_0+|X|)-R(X_0)\big)
 \]
 and
     \[
 p\Big(\big(R(X_0)-R(X_0+X)\big)_+\Big)\leq p\big(R(X_0)-R(X_0-|X|)\big).
 \]
 Since $R$ is monotonic and $p$ is a lattice seminorm, the maps
 \[
  L^0(\Om;B)\to [0,\infty),\quad X\mapsto \begin{cases}
p\big(R(X_0+X)-R(X_0)\big),& X\in L^0(\Om;B_+),\\
  0, & X\notin  L^0(\Om;B_+),
  \end{cases}
 \]
 and
 \[
 L^0(\Om;B)\to [0,\infty),\quad X\mapsto \begin{cases}
  p\big(R(X_0)-R(X_0-X)\big),& X\in  L^0(\Om;B_+),\\
  0, &X\notin  L^0(\Om;B_+),
  \end{cases}
 \]
 are monotonic, and
 the statement follows from Theorem \ref{thm.cont.monotone.L0} a) together with the fact that convergence of a sequence $(X_n)_{n\in \N}\subset L^0(\Om;B)$ to zero in probability implies convergence of the sequence $(|X_n|)_{n\in \N}$ to zero in probability. 
\end{proof}

Let $B$ be a Banach lattice and $L$ be a vector lattice.\ We now specialize to the case where $R\colon B\to L$ is monotonic and positively homogeneous of order $\alpha>0$, i.e.,
\[
R(\la x)=\la^\al R(x)\quad \text{for all }x\in B\text{ and }\la> 0. 
\]
The following corollary is an extension of the well-known Namioka-Klee theorem, cf.\ \cite{Namioka1957}, which states that every positive linear operator from a Banach lattice to a topological vector lattice is bounded, see also \cite[Theorem 4.3, p.\ 182]{AB2006}.

\begin{corollary}\label{cor.cont.poshom}
 Let $B$ be a Banach lattice, $L$ be a vector lattice, and $R\colon B\to L$ be monotonic and positively homogeneous of order $\al>0$.\
Then, for every lattice seminorm $p\colon L\to [0,\infty)$, there exists a constant $C\geq 0$ such that
 \[
 p\big(R(x)\big)\leq C\|x\|^\al\quad \text{for all }x\in B.
 \]
 In particular, $R$ is continuous at the origin w.r.t.\ every lattice seminorm $p\colon L\to [0,\infty)$.
\end{corollary}

\begin{proof}
 Let $p\colon L\to [0,\infty)$ be a lattice seminorm.\ Then, $R(0)=\la^\al R(0)$ for all $\la>0$, so that $R(0)=0$.\ Hence, 
 \[
 \lim_{\la\to 0}p\big(R(\la x)-R(0)\big)=\lim_{\la\to 0}p\big(R(\la x)\big)=\lim_{\la\to 0}\la^\al p\big(R(x)\big)=0\quad \text{for all }x\in B.
 \]
 By Theorem \ref{thm.cont.lattice.L0}, there exists some $\de>0$ such that
 \[
 p(R(x))\leq 1\quad\text{for all }x\in B\text{ with }\|x\|\leq \de.
 \]
 Now, for $C:=\de^{-\al}$, it follows that
 \[
 p\big(R(x)\big)=C\|x\|^\al p\Big(R\big(\tfrac{\de x}{\|x\|}\big)\Big)\leq C\|x\|^\al\quad \text{for all }x\in B\setminus\{0\}.
 \]
 For $x=0$, we have $p\big(R(x)\big)=p(0)=0=C\|x\|^\al$.
\end{proof}

Another direct consequence of Theorem \ref{thm.cont.lattice.L0} is the convex version of the Namioka-Klee theorem, cf. \cite[Corollary 2]{BF2009}.

\begin{corollary}\label{cor.cont.convex}
 Let $B$ be a Banach lattice, $L$ be a vector lattice, and $R\colon B\to L$ be monotonic and convex, i.e.,
\[
R\big(\la x+(1-\la)y\big)\leq \la R(x)+(1-\la)R(y)\quad \text{for all }x,y\in B\text{ and }\la\in (0,1). 
\]
Then, $R$ is continuous w.r.t.\ every lattice seminorm $p\colon L\to [0,\infty)$.
\end{corollary}

\begin{proof}
 Let $p\colon L\to [0,\infty)$ be a lattice seminorm and $x_0\in B$.\ Since $R$ is convex and monotonic and $p$ is a lattice seminorm, the map
 \[
 [0,1]\to [0,\infty),\quad \la\mapsto p\big(R(x_0+\la x)-R(x_0)\big)
 \]
 is convex and therefore continuous at zero for all $x\in B$ with $x\geq 0$.\ Now, the statement follows from Theorem \ref{thm.cont.lattice.L0}.
\end{proof}

\end{document}